%% file: main.tex
\documentclass[12pt, draftclsnofoot,journal,onecolumn]{IEEEtran}
\input src/header.tex

\begin{document}

\title{A Training-Based Mutual Information Lower Bound for Large-Scale Systems}
\author{Xiangbo Meng,~\IEEEmembership{Student Member,~IEEE,} Kang~Gao,~\IEEEmembership{Student Member,~IEEE,} and~Bertrand~M.~Hochwald,~\IEEEmembership{Fellow,~IEEE}%
\thanks{This work was generously supported by NSF Grant \#1731056 and Futurewei Technologies, Inc.}
\thanks{Xiangbo Meng, Kang Gao, and Bertrand M. Hochwald are with the Department of Electrical Engineering, University of Notre Dame, Notre Dame, IN, 46556 USA (email:xmeng@nd.edu; kgao@nd.edu; bhochwald@nd.edu).}}
\maketitle
\thispagestyle{plain}
\pagestyle{plain}
\begin{abstract}
We provide a mutual information lower bound that can be used to analyze the effect of training in models with unknown parameters. For large-scale systems, we show that this bound can be calculated using the difference between two derivatives of a conditional entropy function.  The bound does not require explicit estimation of the unknown parameters.  We provide a step-by-step process for computing the bound, and provide an example application.  A comparison with known classical mutual information bounds is provided.
\end{abstract}
\begin{IEEEkeywords}
information rates, training, entropy, large-scale systems
\end{IEEEkeywords}
\IEEEpeerreviewmaketitle

\section{Introduction}

\comm{
Consider a linear system and get a bound for a system with training. (still require the reference of Taiwanese paper for Entropy computation, but focus on linear systems)
}



Many systems have unknown parameters that are estimated during a training-phase with the help of known prescribed training signals.  This phase is followed by a data phase, where knowledge of the estimated parameters is used to process the data.  It is generally assumed that the parameters are constant during these two phases, the total duration of which is called the coherence time.  It is often of great interest to optimize the training time for a given coherence time, since time in the training phase, while useful for parameter estimation, generally takes away from time in the data phase.

In a communication system, the parameters of interest often include the channel, which is typically unknown and learned at the receiver with the help of pilot signals sent by the transmitter.   For example, \cite{hassibi2003much} analyzes a multi-antenna model and a capacity lower-bound is obtained by using the minimum mean-square estimate (MMSE) of the channel, and the residual channel error is treated as Gaussian noise.  This lower bound is maximized over various parameters, including the fraction of the coherence time that should be dedicated to training.  A similar optimization is considered in \cite{muharar2020optimal}, where the power allocation and training duration are chosen to achieve the maximum sum-rate in a multiuser system.  Such ``one-shot learning", where the parameters are estimated only during the training phase, can be augmented by further refinement during the data phase \cite{takeuchi2012large,takeuchi2013achievable}. However, this refinement can suffer from error propagation \cite{miridakis2016joint}, and we do not consider this herein. 

Many of these previous efforts to analyze training assume that the unknown parameters appear linearly in the system model \cite{hassibi2003much,muharar2020optimal,takeuchi2012large,takeuchi2013achievable,sheng2017optimal}, or appear in a linearized version of the model \cite{li2016much,li2017channel}, often by employing the Bussgang decomposition \cite{bussgang1952crosscorrelation}.  We develop a framework to analyze one-shot training that does not require the parameters to appear linearly in the model, nor does it require additive Gaussian noise; rather, it requires the system to be time-invariant and memoryless, 
and a certain entropy to be computed in the large-scale system limit.
Herein, large-scale refers to long block lengths (time durations) or large dimensional inputs and outputs, or both.
The fact that the large-scale system entropy can sometimes be computed even when the small-scale system entropy cannot is exploited for our training analysis.  

\subsection{Problem setup and statement}

Consider a system model that has input and output processes as $\txA=(\txv_1,\txv_2,\ldots)$ and $\rxA=(\rxv_{1},\rxv_{2},\ldots)$, which comprise vectors $\txv_t$ and $\rxv_{t}$ whose dimensions are $\tn$ and $\rn$ respectively. The input and output are connected through a conditional distribution parameterized by $\chv_{\Tb}$, whose value is unknown.  We assume that $\chv_{\Tb}$ is constant during a coherence time block $\Tb$, and then changes independently in the next block (same length $\Tb$), and so on. 
The system is supplied with known inputs during a ``training phase" to learn the parameters, after which the system is used during its ``data phase".   The unknown parameters are assumed to have a known distribution, and the number of unknown parameters is allowed to be a function of $\Tb$.  {\em Problem statement:} We wish to determine the optimum amount of training.

To analyze the effects of training, a lower bound on the mutual information between the input and output is often used
\begin{align}
    \frac{1}{\Tb}\MuI(\txm_\Tb; \rxm_\Tb) \ge \frac{\Tb - \Tt}{\Tb}\MuI(\txv_{\Tt+1};\rxv_{\Tt+1}|\txm_{\Tt},\rxm_{\Tt}),
    \numberthis
    \label{eq:general_lower_bound}
\end{align}
where $\txv_{t}$ and $\rxv_{t}$ are the $t$th vector input and output of the system,  $\txm_{t}=[\txv_{1},\cdots,\txv_{t}]$, $\rxm_{t}=[\rxv_{1},\cdots,\rxv_{t}]$, and $\Tt$ is the number of training symbols in one coherence block.  We assume $0<\tau<1$ is the fraction of the blocklength devoted to training, and $\Tt$ is integer for convenience.  (We choose this in favor of using $\ceil{\Tt}$ throughout.)  The optimal training fraction, in the sense of maximizing this lower bound, is then
\begin{align}
    \tauopt = \argmax_{\tau}\;(1-\tau)\MuI(\txv_{\Tt+1};\rxv_{\Tt+1}|\txm_{\Tt},\rxm_{\Tt}).
    \numberthis 
    \label{eq:optimal_training_time_finite}
\end{align}
Such an analysis appears, for example, in \cite{hassibi2003much,muharar2020optimal,li2016much,li2017channel}, but the right-hand side of \eqref{eq:optimal_training_time_finite} can be difficult to compute and is itself often approximated or lower bounded.  For example, in \cite{hassibi2003much}, a wireless communication system with Rayleigh block-fading channel and additive Gaussian noise is considered, and the mutual information in \eqref{eq:optimal_training_time_finite} is lower bounded by treating the estimation error of the MMSE estimate of the channel as independent additive Gaussian noise.  However, this form of analysis is often intractable when the parameters appear nonlinearly, or the additive noise is non-Gaussian, since explicit estimates of the unknown parameters are unavailable.


By considering a large-scale limit of the conditional mutual information in \eqref{eq:optimal_training_time_finite}, we provide a method to revisit this computation.  Let $\Tb\to\infty$, and define the ratios
\begin{align}
    \ratio=\frac{\rn}{\tn},\quad \beta=\frac{\Tb}{\tn}.
    \label{eq:ratios}
\end{align}
It is possible, although not required, that $\tn$ and $\rn$ also grow to infinity with $\Tb$, so that $\beta$ is finite.  The large-scale limit of the conditional mutual information $\MuI(\txv_{\Tt+1};\rxv_{\Tt+1}|\txm_{\Tt},\rxm_{\Tt})$ in \eqref{eq:optimal_training_time_finite} is
\begin{align}
    \cIInn{\txA}{\rxA}{}=\lim_{\Tlim\to\infty}\frac{1}{\rn}\MuI(\txv_{\Tt+1};\rxv_{\Tt+1}|\txm_{\Tt},\rxm_{\Tt}).
    \label{eq:MuI_def_high}
\end{align}
The normalization by $1/\rn$ is needed to keep this quantity finite if $\rn\to\infty$, and this limit (assuming that it exists) typically depends on $\ratio$, $\beta$, and $\tau$.  The optimal training time in \eqref{eq:optimal_training_time_finite} then becomes 
\begin{align}
    \tauopt=\argmax_{\tau}\ (1-\tau) \cII{\txA}{\rxA}{},
    \label{eq:opt_training_tau}
\end{align}
and the corresponding optimal rate becomes
\begin{equation*}
    \Ropt = (1-\tauopt)\cII{\txA}{\rxA}{}\bigr|_{\tau=\tauopt}.
\end{equation*}
The value of this analysis depends on our ability to compute $\cII{\txA}{\rxA}{}$, and we show that this quantity can be computed as the derivative of a certain entropy.

\section{Main Results}

\subsection{Assumptions and definitions of useful quantities}
Before we introduce the main results, we first make some assumptions and definitions. The bound in \eqref{eq:general_lower_bound} is fully determined by the distribution of the triple $(\txm_{\Tb},\rxm_{\Tb},\chv_{\Tb})$,  and we make the following assumption:
\begin{align*}
    \text{A1:}\;\;& p(\rxm_{\Tb}|\txm_{\Tb};\chv_{{\Tb}}) =\prod_{t=1}^{\Tb}p(\rxv_t|\txv_t;\chv_{{\Tb}}),
    \numberthis
    \label{eq:fixed_channel_in_data} \\
   & p(\txm_\Tb)=p(\txm_{\Tt})\prod_{t=\Tt+1}^\Tb p(\txv_t),
   \numberthis
    \label{eq:iid_data_input_distribution}
\end{align*}
where $p(\rxv_t|\txv_t;\chv_{\Tb})$ is a fixed conditional distribution for all $t=1,2,\ldots,\Tb$ and $p(\txv_t)$ is a fixed distribution for all $t=\Tt+1,\Tt+2,\ldots,\Tb$.

Equation \eqref{eq:fixed_channel_in_data} says that the system is memoryless and time-invariant (given the input and parameters) and \eqref{eq:iid_data_input_distribution} says that the input $\txv_t$ is \iid and independent of $\txm_\Tt$ for all $t>\Tt$.  
We use the common convention of writing $p(\txm_{\Tt})$ and $p(\txv_t)$ when we mean $p_{\txm_{\Tt}}(\cdot)$ and $p_{\txv_t}(\cdot)$, even though these functions can differ. 
Under A1, the distributions of $(\txm_{\Tb},\rxm_{\Tb},\chv_{{\Tb}})$ are described by the set of known distributions
\begin{equation}
\cP(\Tb,\tau)=\{p(\rxv|\txv;\chv_{{\Tb}})),p(\chv_{{\Tb}})),p(\txm_\Tt),p(\txv_{\Tt+1})\}.
\label{eq:P_set_high}
\end{equation}
These distributions are used to calculate all of the entropies and mutual informations throughout.  The entropies and mutual informations are ``ergodic" in the sense that they are averaged over independent realizations of $\chv_{{\Tb}}$.

Define:
\begin{align*}
\EntInnC{\txA}{\rxA}{\dx}{\dy} =\lim_{\Tlim\to\infty}\frac{1}{\rn}\Ent(\rxv_{{\dy\Tt}+1}|\txm_{{\dx\Tt}},\rxm_{{\dy\Tt}}),
    \numberthis
    \label{eq:EntInn_cond_def_high}
\end{align*}
\begin{align*}
\EntInnC{\txA}{\rxA}{\dx^+}{\dy}
=\lim_{\Tlim\to\infty}&\frac{1}{\rn}\Ent(\rxv_{{\dy\Tt}+1}|\txm_{{\dx\Tt}+1},\rxm_{{\dy\Tt}}),
    \numberthis
    \label{eq:EntInn_cond_def_one_extra_high}
\end{align*}
with $\dx,\dy\in[0,\frac{1}{\tau})$, again assuming these limits exist. Notice that here we treat $\dx\Tt, \dy\Tt$ again as integers to avoid excessive use of the ceiling or floor notation.  We drop the subscripts $\dy$ and $\dx$ in $\EntInnC{\txA}{\rxA}{\dx}{\dy}$ and $\EntInnC{\txA}{\rxA}{\dx^+}{\dy}$ when $\dy=1$ or $\dx=1$.  For example, $\EntInnC{\txA}{\rxA}{}{}$ = $\EntInnC{\txA}{\rxA}{\dx}{\dy}|_{\dx=1, \dy =1}$ and $\EntInnC{\txA}{\rxA}{^+}{}$ = $\EntInnC{\txA}{\rxA}{\dx^+}{\dy}|_{\dx=1, \dy =1}$. 
With \eqref{eq:EntInn_cond_def_high} and \eqref{eq:EntInn_cond_def_one_extra_high}, we further make the following assumptions:
\begin{align*}
\text{A2:}&\qquad\quad \EntInnC{\txA}{\rxA}{^+}{}  = \lim_{\dy\downto 1} \EntInnC{\txA}{\rxA}{\dy^+}{\dy},
    \numberthis
    \label{eq:training_phase_continue}\\
&\qquad\quad \EntInnC{\txA}{\rxA}{}{} =\lim_{\dy\downto 1} \EntInnC{\txA}{\rxA}{}{\dy}.
    \numberthis
    \label{eq:continue_in_data}
\end{align*}


Assumptions A1--A2 in \eqref{eq:fixed_channel_in_data}, \eqref{eq:iid_data_input_distribution},  \eqref{eq:training_phase_continue}, and  \eqref{eq:continue_in_data},  are important for the main result (Theorem \ref{thm:MI_equal_gap&computation_of_MuI_from_derivative}). 
A1 is often met in practice for a memoryless and time-invariant system with \iid input in the data phase, independent of the input and output during training. However, we do not have a complete characterization of the processes $\txA$ and $\rxA$ that meet Assumptions A2.  Nevertheless, A2 may be verified on a case-by-case basis by examining expressions of  $\EntInnC{\txA}{\rxA}{\dy^+}{\dy}$ and $\EntInnC{\txA}{\rxA}{}{\dy}$ with $\dy\geq 1$ using Corollary \ref{cor:A1_derivative_relation}(c) in Appendix \ref{app:proof_MI_equal_gap&computation_of_MuI_from_derivative}; see the Example.


Define
\begin{equation}
    \EntR(\rxA_\dy|\txA_\dx)=\lim_{\Tlim\to\infty}\frac{1}{\rn\Tt}\Ent(\rxm_{{\dy\Tt}}|\txm_{{\dx\Tt}}),
    \label{eq:EntR_cond_def_high}
\end{equation}
\begin{align*}
  \cII{\txA_\dy}{\rxA_\dy}{}= \lim_{\Tlim\to\infty}&\frac{1}{\rn}\MuI(\txv_{{\dy\Tt}+1};\rxv_{{\dy\Tt}+1}|\txm_{{\dy\Tt}},\rxm_{{\dy\Tt}}),
  \numberthis
  \label{eq:cI_def_with_P_tau_high}
\end{align*}
\begin{align*}
    \cIRA{\txA_\dy}{\rxA_\dy}{}=\lim_{\Tlim\to\infty}&\frac{1}{\rn\Tb}\MuI(\txm^{{\dy\Tt}+1};\rxm^{{\dy\Tt}+1}|\txm_{{\dy\Tt}},\rxm_{{\dy\Tt}}),
    \numberthis
    \label{eq:joint_MuI_thm_high}
\end{align*}
where $\txm^t = [\txv_t,\txv_{t+1},\cdots,\txv_{\Tb}]$, and $\rxm^t=[\rxv_t,\rxv_{t+1},\cdots,\rxv_{\Tb}]$.  Similarly, we drop the subscripts $\dy$ and $\dx$ in $\cII{\txA_\dy}{\rxA_\dy}{}$, $\cIRA{\txA_\dy}{\rxA_\dy}{}$,  and $\EntR(\rxA_\dy|\txA_\dx)$ when $\dy=1$ or $\dx=1$.

\subsection{Main result}
\begin{thm}
\label{thm:MI_equal_gap&computation_of_MuI_from_derivative}
Under Assumption A1, 
\begin{align*}
    \cIRA{\txA_{0}}{\rxA_{0}}{}&\geq\;(1-\tau)\cII{\txA}{\rxA}{}
    \numberthis
    \label{eq:MuI_begin_lower_bound}\\
    &\geq\; (1-\tau)\lim_{\Tlim\to\infty}\MuI(\txv_{\Tt+1};\rxv_{ \Tt+1}|\hat{\chv}_{\Tb}),
    \numberthis
    \label{eq:MuI_estimate_lower_bound}
\end{align*}
where $\hat{\chv}_{\Tb}$ is any estimate of $\chv_{{\Tb}}$ that is a function of $(\txm_{ \Tt},\rxm_{ \Tt})$.  When A2 is also met,
\begin{equation}
    \cIInn{\txA}{\rxA}{} = \lim_{\dy\downto 1}\frac{\partial\EntR(\rxA_\dy|\txA)}{\partial\dy} - \lim_{\dy\downto 1} \frac{\partial\EntR(\rxA_\dy|\txA_\dy)}{\partial\dy}.
    \label{eq:MuI_from_derivative}
\end{equation}

\end{thm}

\begin{proof}
Please see Appendix \ref{app:proof_MI_equal_gap&computation_of_MuI_from_derivative}.
\end{proof}

Hence, the mutual information limit with one-shot learning $(1-\tau)\cIInn{\txA}{\rxA}{}$ is a lower bound of the mutual information without any training, and is an upper bound of the mutual information with any estimate of the unknown parameters.  The expression in \eqref{eq:MuI_begin_lower_bound} can be calculated as a derivative using \eqref{eq:MuI_from_derivative} as long as $\EntR(\rxA_\dy|\txA_\dx)$ is available.
The next section shows how $\EntR(\rxA_\dy|\txA_\dx)$ may be computed, and the section that follows provides operational significance to \eqref{eq:MuI_begin_lower_bound} in the form of a channel coding theorem.  An example application of the theorem appears in Section~\ref{sec:step_by_step_tauopt}.

\subsection{Computation of $\EntR(\rxA_\dy|\txA_\dx)$}

An expression for $\EntR(\rxA_\dy|\txA_\dx)$ may be derived from $\EntR(\rxA_\dy|\txA)$ when this latter quantity is available.  In some cases $\EntR(\rxA_\dy|\txA)$ can be obtained through methods employed in statistical mechanics by treating the conditional entropy as free energy in a large-scale system.  Free energy is a fundamental quantity \cite{castellani2005spin,mezard2009information} that has been analyzed through the powerful ``replica method", and this, in turn, has been applied to entropy calculations in 
machine learning \cite{engel2001statistical,opper1996statistical,shinzato2008learning,ha1993generalization} and wireless communications \cite{wen2015performance,wen2015joint,wen2016bayes}, in both linear and nonlinear systems.  

The entropy $\EntR(\rxA|\txA)$ (equivalent to $\EntR(\rxA_\dy|\txA)$ where $\dy=1$) is considered in \cite{engel2001statistical,opper1996statistical,ha1993generalization,shinzato2008learning}, where the input is multiplied by an unknown vector as an inner product and then passes through a nonlinearity to generate a scalar output. In \cite{engel2001statistical,opper1996statistical,ha1993generalization}, the inputs are \iid, while orthogonal inputs are considered in \cite{shinzato2008learning}. 
The entropy $\EntR(\rxA_{\dy}|\txA)$ for MIMO systems is considered in \cite{wen2015performance,wen2015joint,wen2016bayes}, where the inputs are \iid in the training phase and are \iid in the data phase, but the distributions in the two phases can differ. In \cite{wen2015performance}, a linear system is considered where the output is the result of the input multiplied by an unknown matrix, plus additive noise, while in \cite{wen2015joint,wen2016bayes} uniform quantization is added at the output.

As we now show, the expression for $\EntR(\rxA_{\dy}|\txA)$ for $\dy\geq 1$ can be leveraged to compute $\EntR(\rxA_\dy|\txA_\dx)$ for all $\dy,\dx>0$.  We consider the case when the input $\txv_t$ are \iid for all $t$, and the distribution set $\cP(\Tb,\tau)$ defined in \eqref{eq:P_set_high} can therefore be simplified as
\begin{align}
    \cP(\Tb,\tau)=\{p(\rxv|\txv;\chva{\Tb}),p(\chva{\Tb}),p(\txv)\}.
    \label{eq:P_set_high_iid_input}
\end{align}
The following theorem assumes that we have $\EntR(\rxA_\dy|\txA)$ available as a function of $(\tau,\dy)$ for all $\dy\geq 1$.

\begin{thm}
\label{thm:scale_joint_entropy}
Assume that Assumption A1 is met, $\txv_t$ are \iid for all $t$, $\EntR(\rxA_\dy|\txA)$ exists and is continuous in $\tau$ and $\dy$ for $\tau\in(0,1)$ and $\dy\in(0,\frac{1}{\tau}]$. Define
\begin{align}
    F(\tau,\dy)=\EntR(\rxA_\dy|\txA),
    \label{eq:entropy_generator}
\end{align}
where $\dy\geq 1$ and $\EntR(\rxA_\dy|\txA)$ is defined in \eqref{eq:EntR_cond_def_high}.
Then
\begin{align*}
    \EntR(\rxA_\dy|\txA_{\dx})=
      u\cdot F\left(u\tau,\frac{\dy-u}{\dx}+1\right), 
    \numberthis
    \label{eq:EntR_rx_gamma_tx_lambda}
\end{align*}
for all $\dy,\dx\in(0,\frac{1}{\tau}]$, where $u=\min(\dy,\dx)$.
\end{thm}
\begin{proof}
According to  \eqref{eq:EntR_cond_def_high} and \eqref{eq:entropy_generator}, we have
\begin{align*}
    F(\tau,\dy)=\lim_{\Tlim\to\infty}\frac{1}{\rn\Tt}\Ent(\rxm_{{\dy\Tt}}|\txm_{\Tt}),
\end{align*}
which is computed using $\cP(\Tb,\tau)$ defined in \eqref{eq:P_set_high_iid_input}.  When $\dx\geq\dy>0$, we have
\begin{align}
    &\EntR(\rxA_\dy|\txA_\dx)=\lim_{\Tlim\to\infty}\frac{1}{\rn\Tt}\Ent(\rxm_{{\dy\Tt}}|\txm_{{\dx\Tt}}) \label{eq:Hdef}\\
    =&\lim_{\Tlim\to\infty}\frac{1}{\rn\Tt}\Ent(\rxm_{{\dy\Tt}}|\txm_{{\dy\Tt}})=\lim_{{\Tlim}\to\infty}\frac{\dy}{\rn\tilde{\tau}{\Tb}}\Ent(\rxm_{\tilde{\tau}{\Tb}}|\txm_{\tilde{\tau}{\Tb}}).
    \label{eq:change_of_variable}
\end{align}
where $\tilde{\tau}={\dy\tau}$.
Therefore, \eqref{eq:entropy_generator} and \eqref{eq:change_of_variable} yield
\begin{align*}
   \EntR(\rxA_\dy|\txA_\dx) = \dy\cdot F(\tilde{\tau},1)=\dy\cdot F(\dy\tau,1).
   \numberthis
   \label{eq:expand_1}
\end{align*}
When $\dy>\dx> 0$, let $\tilde{\tau}={\dx\tau}$, and then \eqref{eq:Hdef} yields
\begin{align*}
   \EntR(\rxA_\dy|\txA_\dx) &=\lim_{{\Tlim}\to\infty}\frac{\dx}{\rn\tilde{\tau}{\Tb}}\Ent(\rxm_{\dy\tilde{\tau}{\Tb}/\dx}|\txm_{\tilde{\tau}{\Tb}})\\
    &=\dx\cdot F\left(\tilde{\tau},\frac{\dy}{\dx}\right) = \dx\cdot F\left(\dx{\tau},\frac{\dy}{\dx}\right) .
    \numberthis
   \label{eq:expand_2}
\end{align*}
By combining \eqref{eq:expand_1} and \eqref{eq:expand_2}, we obtain \eqref{eq:EntR_rx_gamma_tx_lambda}.
\end{proof}

\subsection{Channel coding theorem}
\label{subsec:channel_coding_thm_SISO}
We now provide an operational description of the mutual information inequality \eqref{eq:MuI_begin_lower_bound}. 
We consider a communication system where the channel is constant for blocklength $\Tb$, and then changes independently and stays constant for another blocklength, and so on. The first $\Tt$ symbols of each block are used for training with known input and output.
Under Assumption A1, the communication system is memoryless, is time-invariant within each block, and the input is \iid independent of $\txm_{\Tt}$ after training.  The system is retrained with every block, and the message to be transmitted is encoded over the data phase of multiple blocks. 

A $(2^{nR\rn\Tb },n,\Tb)$-code for a block-constant channel with blocklength $\Tb$ is defined as an encoder that maps a message $S\in\{1,2,\ldots,2^{nR\rn\Tb}\}$ to the input in the data phase $\txm^{\Tt+1}$ among $n$ blocks,
and a decoder that maps $\txm_{\Tt}$, and the entire output $\rxm_{\Tb}$ for $n$ blocks to $\hat{S}\in\{1,2,\ldots,2^{nR\rn\Tb}\}$, where $\rn = \frac{\alpha\Tb}{\beta}$. The code rate $R$ has units ``bits per transmission per receiver", and the maximum probability of error of the code is defined as
\begin{equation}
    \Pe(n,\Tb)=\max_{S}\Pr(\hat{S}\neq S).
    \label{eq:max_prob_err}
\end{equation}
The channel coding theorem is shown below.
\begin{thm}
\label{thm:channel_coding_thm_inf_Tb}
Assume A1 is met, with a channel that is constant with blocklength $\Tb$, whose conditional distribution is parameterized by $\chva{\Tb}$ and is independent of the input.  If $\cII{\txA}{\rxA}{}$ exists, then for every $R$ that satisfies
\begin{equation*}
    R<(1-\tau)\cII{\txA}{\rxA}{},
\end{equation*}
there exists $\Tb_0>0$, so that for all $\Tb>\Tb_0$, we can find a code $(2^{nR\rn\Tb },n,\Tb)$ with maximum probability of error $\Pe(n,\Tb)\to 0$ as $n\to\infty$.
\end{thm}
\begin{proof}
Define 
\begin{equation*}
    \cR_{\Tb}=\frac{1}{\Tb\rn}\MuI(\txm^{\Tt+1};\rxm_{\Tb}|\txm_{\Tt}).
\end{equation*}
For any finite $\Tb$, according to the classical channel coding theorem \cite{cover2012elements,yeung2008information,effros2010generalizing}, for every $R<\cR_{\Tb}$, there exists a code $(2^{n R \rn \Tb},n,\Tb)$ with maximum probability of error $\Pe(n,\Tb)\to 0$ as $n\tendsto\infty$. 

It is clear that $\txm^{\Tt+1}$ is independent of $(\txm_{\Tt},\rxm_{\Tt})$. Therefore, we have
\begin{equation*}
    \cR_{\Tb}=\frac{1}{\Tb\rn}\MuI(\txm^{\Tt+1};\rxm^{\Tt+1}|\txm_{\Tt},\rxm_{\Tt}).
\end{equation*}

Since $\txv_{\Tt+1},\txv_{\Tt+2},\ldots,\txv_{\Tb}$ are \iid, and $p(\rxv_t|\txv_t;\chva{\Tb})$ is a fixed conditional distribution for all $t=1,2,\ldots,\Tb$, we have 
\begin{equation}
    \cR_{\Tb}\geq \frac{(1-\tau)}{\rn}\MuI(\txv_{\Tt+1};\rxv_{\Tt+1}|\txm_{\Tt},\rxm_{\Tt}).
    \label{eq:bound_of_R}
\end{equation}
According to the definition in \eqref{eq:cI_def_with_P_tau_high},
\begin{align*}
     \cII{\txA}{\rxA}{}=\lim_{\Tlim\to\infty} \frac{1}{\rn} \MuI(\txv_{\Tt+1};\rxv_{\Tt+1}|\txm_{\Tt},\rxm_{\Tt}).
\end{align*}
Therefore, for any $\kappa>0$, there exists a number $\Tb_0>0$ so that when $\Tb>\Tb_0$, we have
\begin{align*}
    \frac{1}{\rn}\MuI(\txv_{\Tt+1};\rxv_{\Tt+1}|\txm_{\Tt},\rxm_{\Tt})>\cII{\txA}{\rxA}{} - \kappa,
\end{align*}
and \eqref{eq:bound_of_R} yields
\begin{align*}
    \cR_{\Tb}>(1-\tau)(\cII{\txA}{\rxA}{} - \kappa),
\end{align*}
which means any rate $R\leq(1-\tau)(\cII{\txA}{\rxA}{} - \kappa)$ is achievable. 

By taking the limit $\kappa\downto 0$, we finish the proof.
\end{proof} 

This theorem shows that rates below $(1-\tau)\cII{\txA}{\rxA}{}$ are achievable when $\Tb$ is chosen large enough.  Only an achievability statement is given here since $(1-\tau)\cII{\txA}{\rxA}{}$ is a lower bound on $\cR_{\Tb}$ for large $\Tb$. 

\section{Steps for Computing Optimal Training Time and an Example}
\label{sec:step_by_step_tauopt}
We summarize the process to compute the optimal training time $\tauopt$ for a memoryless, time-invariant system with unknown parameters.  We assume that the input dimension $\tn$, the output dimension $\rn$, and the coherence time (block of symbols) $\Tb$ have the ratios defined in \eqref{eq:ratios}. The unknown parameters of the system are constant within the block, and change independently in the next block.  The first $\Tt$ symbols of each block are used for training and the remaining $\Tb-\Tt$ are for data.  We assume $\Tlim\to \infty$, and solve \eqref{eq:opt_training_tau} as an approximation of \eqref{eq:optimal_training_time_finite}.  The input $\txv_t$ are \iid for all $t=1,\dots,\Tb$. 

The process includes the following seven steps:
\begin{itemize}
    \item [1)] Verify Assumption A1 \eqref{eq:fixed_channel_in_data}--\eqref{eq:iid_data_input_distribution} based on the set of distributions in \eqref{eq:P_set_high}.
    
    \item [2)] Compute $\EntR(\rxA_\dy|\txA)$ defined in \eqref{eq:EntR_cond_def_high} for $\dy \ge 1$ based on \eqref{eq:P_set_high}, and express it as a function of $\tau$ and $\dy$ as $F(\tau,\dy)$ \eqref{eq:entropy_generator}.
    
    \item [3)] Compute $\EntR(\rxA_\dy|\txA_\dx)$ defined in \eqref{eq:EntR_cond_def_high}, for all $\dy, \dx \in (0, \frac{1}{\tau}]$ by using Theorem~\ref{thm:scale_joint_entropy} and $F(\tau,\dy)$.
    
    \item [4)] Compute $\EntInnC{\txA}{\rxA}{\dx}{\dy}$ and $\EntInnC{\txA}{\rxA}{\dy^+}{\dy}$ defined in \eqref{eq:EntInn_cond_def_high}-\eqref{eq:EntInn_cond_def_one_extra_high} by taking the derivative of $\EntR(\rxA_\dy|\txA_\dx)$ and $\EntR(\rxA_\dy|\txA_\dy)$ (Corollary \ref{cor:A1_derivative_relation}(a)--(b) in Appendix).
    
    \item [5)] Verify Assumption A2 \eqref{eq:training_phase_continue} by examining the expressions of $\EntInnC{\txA}{\rxA}{\dy^+}{\dy}$ and verify \eqref{eq:continue_in_data} with Corollary \ref{cor:A1_derivative_relation}(c).
    
    \item [6)] Compute $\cIInn{\txA}{\rxA}{}$ by using \eqref{eq:MuI_from_derivative}.

    \item [7)] Solve $\tauopt$ using \eqref{eq:opt_training_tau}.
    
\end{itemize}
    

    
    
    


The following simple example applies these steps.
\subsection*{Example: Bit flipping through random channels}
\noindent
Let 
\begin{equation}
    \rx_{t}=\tx_{t}\oplus \chs_{k_t},\quad t=1,\ldots,\Tb,
    \label{eq:binary_XOR_channel}
\end{equation}
where, since $\tn=\rn=1$, the binary input $\tx_t$ and output $\rx_t$ are scalars, and $\tx_{t}$ is XOR'ed with a random bit $\chs_{k_t}$ intended to model the unknown ``state" of the channel $k_t$.  Thus, each channel either lets the input bit directly through, or inverts it.  The $\tx_{t}$ are \iid equally likely to be zero or one, Bernoulli($\frac{1}{2}$) random variables.  Let $a>0$ be a parameter, where $a\cdot\Tb$ is the (integer) number of unique channels whose states are stored in the vector $\chv_{\Tb}=[\chs_{1},\chs_{2},\cdots,\chs_{a\cdot\Tb}]^\Tp$ comprising \iid Bernoulli($\frac{1}{2}$) random variables that are independent of the input. The channel selections $\bk_{\Tb}=[k_1,\ldots,k_{\Tb}]$ are chosen as an \iid uniform sample from $\{1,2,\cdots,a\cdot\Tb\}$ (with possible repetitions), and the choices are known to the receiver.  We wish to send training signals through these channels to learn $\chv_{\Tb}$; the more entries of this vector that we learn, the more channels become useful for sending data, but the less time we have to send data before the blocklength $\Tb$ runs out and $\chv_{\Tb}$ changes.
We want to determine the optimum $\tau$ as $\Tlim\to\infty$ using \eqref{eq:opt_training_tau}. We therefore follow the steps above.

\begin{itemize}

\item [1)] From \eqref{eq:binary_XOR_channel}, we have
\begin{align*}
    p(\rxv_{\Tb}|\txv_{\Tb};\chva{\Tb}) =\prod_{t=1}^{\Tb}p(\rx_t|\tx_t;\chva{\Tb}),
\end{align*}
where $p(\rx_t|\tx_t;\chva{\Tb})=\mathbbm{1}_{(\rx_t=\tx_{t}\oplus \chs_{k_t})}$ for all $t$.  Here the notation is slightly abused, since now $\tn=\rn=1$, we use $\rxv_{t}$ and $\txv_{t}$ to denote $[\rx_1,\dots,\rx_{t}]^\Tp$ and $[\tx_1,\dots,\tx_{t}]^\Tp$.
It is clear that Assumption A1 is met and $\tx_t$ are \iid for all $t$ independent of $\chv_{\Tb}$.

\item [2)]
By definition, $\EntR(\rxA|\txA)=\lim\limits_{\Tlim\to\infty}\frac{1}{\Tt}\Ent(\rxv_{\Tt}|\txv_{\Tt})$. 
The model \eqref{eq:binary_XOR_channel} yields
\begin{align*}
    \Ent(\rxv_{\Tt}|\txv_{\Tt}) &\stackrel{(\rm a)}{=}\Ent(\{\chs_{k_1},\ldots,\chs_{k_\Tt}\}|\txv_{\Tt})
    \stackrel{(\rm b)}{=}\Ent(\{\chs_{k_1},\ldots,\chs_{k_\Tt}\})
    \stackrel{(\rm c)}{=}\E_{\bk_{\Tt}}|A_\Tt| \\
    &\stackrel{(\rm d)}{=}\sum_{i=1}^{a\Tb}\E(\mathbbm{1}_{(i\in A_{\Tt})})   ={a\Tb}(1-(1-\frac{1}{a\Tb})^{\Tt}).
\end{align*}
where $A_{\Tt} = \{k_1,\ldots,k_\Tt\}$, $^{(a)}$ uses $\chs_{k_t}=\rx_t\oplus\tx_t$, $^{(b)}$ uses the independence between $\txv_{\Tt}$ and $\chs_t$, $^{(c)}$ uses the independence between $\chs_t$, $\chs_k$ when $t\neq k$ and $^{(d)}$ uses $|A_{\Tt}|=\sum_{i=1}^{a\Tb} \mathbbm{1}_{(i\in A_{\Tt})}$, where $\mathbbm{1}_{(\cdot)}$ is the indicator function.
Therefore,
\begin{align*}
    &\EntR(\rxA|\txA)=\lim_{\Tlim\to\infty}\frac{a}{\tau}(1-(1-\frac{1}{a\Tb})^{\Tt})=\frac{a}{\tau}(1-e^{-\frac{\tau}{a}}).
    \numberthis
    \label{eq:XOR_tau_1}
\end{align*}
By the chain rule for entropy, for $\dy > 1$, we have
\begin{align*}
    \EntR&(\rxA_\dy|\txA)=\EntR(\rxA|\txA) + \lim_{\Tlim\to\infty}\frac{1}{\Tt}\Ent(\rx_{\Tt+1},\rx_{\Tt+2},\ldots,\rx_{{\dy\Tt}}|\txv_{\Tt},\rxv_{\Tt}).
\end{align*}
Since
\begin{align*}
   {\dy\Tt} - \Tt\geq&\Ent(\rx_{\Tt+1},\rx_{\Tt+2},\ldots,\rx_{{\dy\Tt}}|\txv_{\Tt},\rxv_{\Tt})\\
    \ge& \Ent(\tx_{\Tt+1},\tx_{\Tt+2},\ldots,\tx_{{\dy\Tt}}|\txv_{\Tt},\rxv_{\Tt},\chva{\Tb})\\
    =&\Ent(\tx_{\Tt+1},\tx_{\Tt+2},\ldots,\tx_{{\dy\Tt}})={\dy\Tt} - \Tt,
\end{align*}
we conclude that 
\begin{align*}
    F(\tau,\dy)=\EntR&(\rxA_\dy|\txA)=\frac{a}{\tau}(1-e^{-\frac{\tau}{a}})+\dy-1.
    \numberthis
    \label{eq:tau_XOR_EntR}
\end{align*}

\item [3)]
Theorem \ref{thm:scale_joint_entropy} yields
\begin{align*}
    \EntR(\rxA_\dy|\txA_\dx) =  \begin{cases} 
      \frac{a}{\tau}(1-e^{-\frac{\tau}{a}\dy}), & \dy\leq\dx; \\
      \frac{a}{\tau}(1-e^{-\frac{\tau}{a}\dx})+(\dy-\dx), & \dx<\dy,
      \end{cases}
\end{align*}
for $\dy,\dx\in(0,\frac{1}{\tau})$.

\item[4)]
Then, Corollary \ref{cor:A1_derivative_relation}(a)-(b) yields
\begin{align*}
    \EntInnC{\txA}{\rxA}{\dy^+}{\dy} =\frac{\partial\EntR(\rxA_\dy|\txA_\dy)}{\partial\dy}=e^{-\frac{\tau}{a}\dy},
\end{align*}
\begin{align*}
    \EntInnC{\txA}{\rxA}{\dx}{\dy} =\frac{\partial\EntR(\rxA_\dy|\txA_\dx)}{\partial\dy}=  \begin{cases}
      e^{-\frac{\tau}{a}\dy}, & \dy<\dx; \\
      1, & \dy>\dx.
      \end{cases}
\end{align*}

\item[5)]
$\EntInnC{\txA}{\rxA}{\dy^+}{\dy}$ and Corollary \ref{cor:A1_derivative_relation}(c) allow us to conclude that Assumption A2 also holds.

\item[6)]
From Theorem \ref{thm:MI_equal_gap&computation_of_MuI_from_derivative}, we obtain
\begin{align*}
    \cII{\txA}{\rxA}{}
    =&\;1-e^{-\frac{\tau}{a}},
\end{align*}
\item [7)]
Finally, \eqref{eq:opt_training_tau} yields \begin{equation}
    \tauopt=\argmax\limits_{\tau}(1-\tau)(1-e^{-\frac{\tau}{a}}),
    \label{eq:Ex5tauopt_def}
\end{equation}
or
\begin{equation*}
    \tauopt=\begin{cases} 
      -a\ln a, & a\to 0; \\
      \frac{1}{2}, & a\to\infty;\\
      \frac{1}{e}, & a=\frac{1}{e}.
      \end{cases}
\end{equation*}
\end{itemize}
When $a$ is small, $\tauopt$ is larger than $a$; when $a$ is large, $\tauopt$ saturates at $\frac{1}{2}$; and $a=\frac{1}{e}$ is the dividing line between $\tauopt>a$ and $\tauopt<a$.  The corresponding rates are
\begin{equation*}
    \Ropt=\begin{cases} 
      (1+a\ln a)(1-a), & a\to 0; \\
      \frac{1}{2}(1-e^{-\frac{1}{2a}}), & a\to\infty;\\
      (1-\frac{1}{e})^2, & a=\frac{1}{e}.
      \end{cases}
\end{equation*}
The optimum fraction of the blocklength $\Tb$ that should be devoted to training varies as a function of the number of possible unique channels.  When $a=1$, the number of unique channels equals $\Tb$, and the $\tauopt\approx 0.44$.  For a large number of unique channels relative to the blocklength ($a\tendsto\infty$), the fraction of the training time saturates at $1/2$.  When $a$ is small, the optimum fraction of the blocklength devoted to training decreases to zero, but more slowly than $a$.

In this example, a traditional finite-system information-theoretic analysis and simulation is possible (these calculations are omitted). Figure \ref{fig:example_compare_triditional_and_derivative} shows the results as a function of $\Tb$, where we can see that as $\Tb$ grows, the resulting $\tauopt$ quickly approaches the large-system results.  The fact that we can use a large-system limit to approximate a finite-system limit is important when applying the Theorems in realistic scenarios.

\begin{figure}
\includegraphics[width=3.8in]{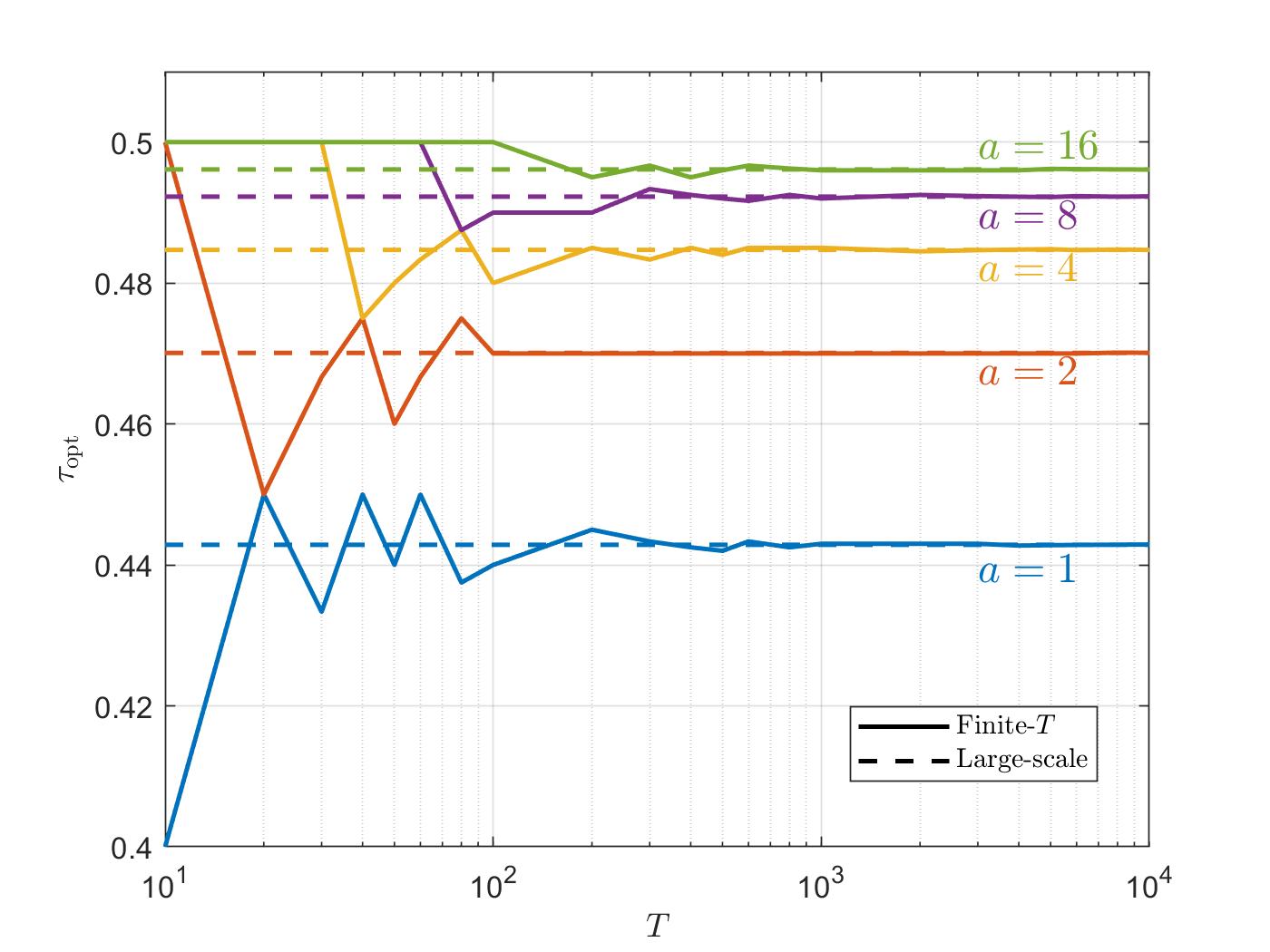}
\centering
    \caption{Optimal training time $\tauopt$ vs blocklength $\Tb$, derived from finite-$\Tb$ analysis (solid curves) and the proposed large-scale method (dashed curves).  The large-scale analysis \eqref{eq:Ex5tauopt_def} shows excellent agreement with the finite-dimensional analysis (which is omitted) for even small values of $\Tb$.}
    \label{fig:example_compare_triditional_and_derivative}
\end{figure}


\section{Discussion and Conclusion}
\label{sec:discussion}

\subsection{Number of unknowns and bilinear model}
In general, a finite number of unknowns in the model leads to uninteresting results as $\Tb\tendsto\infty$.  For example, consider a system modeled as
\begin{equation*}
    \rx_{t} = \chs\tx_t +v_t,\quad t=1,2,\ldots
\end{equation*}
where $\chs$ is the unknown gain of the system, $\tx_t$, $\rx_t$ are the input and corresponding output, $v_t$ is the additive noise, $\tau$ is the fraction of time used for training.  This system is bilinear in the gain and the input.  We assume that $\ns_t$ is modeled as \iid Gaussian $\cN(0,1)$, independent of the input. The training signals are $\tx_t=1$ for all $t=1,2,\ldots, \Tt$, and the data signals $\tx_t$ are modeled as \iid Gaussian $\cN(0,1)$ for all $t=\Tt+1,\Tt+2,\ldots$ An analysis similar to the Example produces
\begin{align*}
    \cIRA{\txA_0}{\rxA_0}{} \geq\frac{1-\tau}{2}\E_{\chs}\log(1+\chs^2),
\end{align*}
and therefore $\tauopt=0$ maximizes this bound. This result reflects the fact that $\chs$ is learned perfectly for any $\tau>0$ because there is only one unknown parameter for $\Tt$ training symbols as $\Tt\tendsto\infty$.  Hence, trivially, it is advantageous to make $\tau$ as small as possible.

More interesting is the ``large-scale" model
\begin{align*}
    \rxv_t=\bff(\chm\txv_t+\nv_t),\quad t=1,2,\ldots,
    \numberthis
    \label{eq:nonlinear_system_model}
\end{align*}
where $\txv_t$ and $\rxv_t$ are the $t$th input and output vectors with dimension $\tn$ and $\rn$, $\chm$ is an $\rn\times\tn$ unknown random matrix that is not a function of $t$, $\nv_1,\nv_2,\ldots$ are \iid unknown vectors with dimension $\rn$ and known distribution (not necessarily Gaussian), and $\bff(\cdot)$ applies a possibly nonlinear function $f(\cdot)$ to each element of its input. The training interval $\Tt$ is used to learn $\chm$.  

Let $\tn$ and $\rn$ increase proportionally to the blocklength $\Tb$ with the ratios defined in \eqref{eq:ratios};
such a model can be used in large-scale wireless communication, signal processing, and machine learning applications. In wireless communication and signal processing \cite{takeuchi2010achievable,takeuchi2012large,hassibi2003much,takeuchi2013achievable,li2016much,li2017channel,wen2015performance,wen2015joint,wen2016bayes}, $\txv_t$ and $\rxv_t$ are the transmitted signal and the received signal at time $t$ in a multiple-input-multiple-output (MIMO) system with $\tn$ transmitters and $\rn$ receivers, $\chm$ models the channel coefficients between the transmitters and receivers, $\Tb$ is the coherence time during which the channel $\chm$ is constant, $\nv_t$ is the additive noise at time $t$, $f(\cdot)$ models receiver effects such as quantization in analog-to-digital converters and nonlinearities in amplifiers. A linear receiver, $f(x)=x$, is considered in \cite{takeuchi2010achievable,takeuchi2012large,takeuchi2013achievable,wen2015performance}. Single-bit ADC's with $f(x)=\sign(x)$ are considered in \cite{li2016much,li2017channel}, and low-resolution ADC's with $f(x)$ modeled as a uniform quantizer are considered in \cite{wen2015joint,wen2016bayes}.  The training and data signals can be chosen from different distributions, as in \cite{hassibi2003much,li2016much,li2017channel}.  Conversely, the training and data signals can both be \iid, as in \cite{takeuchi2013achievable,wen2015performance,wen2015joint,wen2016bayes}.


Let $\rn=1$.  In machine learning, \eqref{eq:nonlinear_system_model} is a model of a single layer neural network (perceptron) \cite{opper1996statistical,engel2001statistical,shinzato2008learning} and $\txv_t$ is the input to the perceptron with dimension $\tn$, $\rxv_t$ is the scalar decision variable at time $t$, $\chm$ holds the unknown weights of the perceptron, and $f(\cdot)$ is the nonlinear activation function. A perceptron is often used as a classifier, where the output of the perceptron is the class label of the corresponding input. In \cite{opper1996statistical,engel2001statistical}, \iid inputs are used to learn the weights, and orthogonal inputs are used in \cite{shinzato2008learning}. Binary class classifiers are considered in \cite{opper1996statistical,engel2001statistical,shinzato2008learning}.
Training employs $\Tt$ labeled input-output pairs $(\txv_t,\rxv_t)$, and the trained perceptron then classifies new inputs before it is retrained on a new dataset.  Generally, both the training and data are modeled as having the same distribution.

To obtain optimal training results for \eqref{eq:nonlinear_system_model}, Theorems \ref{thm:MI_equal_gap&computation_of_MuI_from_derivative}--\ref{thm:scale_joint_entropy} show that a starting point for computing
$\cII{\txA}{\rxA}{}$ is $\EntR(\rxA_\dy|\txA)$ for $\dy\geq 1$.  Fortunately, $\EntR(\rxA_\dy|\txA)$ results can sometimes be found in the existing literature; for example, in \cite{wen2015performance,wen2015joint,wen2016bayes}, $\EntR(\rxA_\dy|\txA)$ is used to calculate the mean-square error of the estimated input signal, conditioned on the training.  Our analysis shows how to leverage these same $\EntR(\rxA_\dy|\txA)$ results to derive the training-based mutual information.  

\subsection{Models for which assumptions are superfluous}

Assumption A2 is likely to be superfluous for certain common system models, such as when the distribution on $\txv_t$ is \iid through the training and data phases, and the transition probabilities can be written as a product as in Assumption A1. However, we have not yet characterized for which models A2 is automatically satisfied without additional assumptions on ${\cal H}'$, and think that this would be an interesting research topic for further work.

\appendices

\section{Proof of Theorem \ref{thm:MI_equal_gap&computation_of_MuI_from_derivative}}
\label{app:proof_MI_equal_gap&computation_of_MuI_from_derivative}

\subsection{Proof of the inequalities \eqref{eq:MuI_begin_lower_bound}--\eqref{eq:MuI_estimate_lower_bound}}
Under Assumption A1, we have
\begin{align*}
    &\MuI(\txm^{\Tt+1};\rxm_{\Tb}|\txm_\Tt)\\
    =&\MuI(\txm^{\Tt+1};\rxm_{\Tt}|\txm_\Tt)+\MuI(\txm^{\Tt+1};\rxm^{\Tt+1}|\txm_\Tt,\rxm_{\Tt})\\
    =&\MuI(\txm^{\Tt+1};\rxm^{\Tt+1}|\txm_\Tt,\rxm_{\Tt}),
\end{align*}
where the first equality uses the chain rule and the second uses that $\txm^{\Tt+1}$ is independent of $(\txm_\Tt,\rxm_\Tt)$. 
Moreover, 
\begin{align*}
    &\MuI(\txm^{\Tt+1};\rxm^{\Tt+1}|\txm_{\Tt},\rxm_{\Tt})\\
    \stackrel{}{=}&\;\Ent(\txm^{\Tt+1}|\txm_{\Tt},\rxm_{\Tt}) - \Ent(\txm^{\Tt+1}|\txm_{\Tt},\rxm_{\Tb})\\
    \stackrel{\rm{( a)}}{=}&\sum_{t={\Tt+1}}^{\Tb}(\Ent(\txv_t|\txm_{t-1},\rxm_{\Tt}) - \Ent(\txv_t|\txm_{t-1},\rxm_{\Tb}))\\
    \stackrel{\rm{( b)}}{\geq}&\sum_{t={\Tt+1}}^{\Tb}(\Ent(\txv_t|\txm_{t-1},\rxm_{t-1}) - \Ent(\txv_t|\txm_{t-1},\rxm_{t}))\\
    =& \sum_{t=\Tt+1}^{\Tb}\MuI(\txv_t;\rxv_t|\txm_{t-1},\rxm_{t-1}).
\end{align*}
Here, $^{(a)}$ uses the chain rule, $^{(b)}$ uses conditioning to reduce entropy. Equality in $^{(b)}$ can be achieved when $\chva{\Tb}$ is estimated perfectly from $(\txm_{\Tt},\rxm_{\Tt})$. 
Since Assumption A1 implies that $\txv_t$ is independent of $(\txv_k,\rxv_k)$ when $k\neq t$ and $t\geq \Tt+1$, for all $t\geq \Tt+1$, we have
\begin{align*}
    &\MuI(\txv_{t+1};\rxv_{t+1}|\txm_{t},\rxm_{t}) - \MuI(\txv_{t};\rxv_{t}|\txm_{t-1},\rxm_{t-1})\\
    \stackrel{\rm{( a)}}{=}&(\Ent(\txv_{t+1}) - \Ent(\txv_{t+1}|\txm_t,\rxm_{t+1})) \\
    &\qquad - (\Ent(\txv_{t}) - \Ent(\txv_{t}|\txm_{t-1},\rxm_{t}))\\
    =&\Ent(\txv_{t}|\txm_{t-1},\rxm_{t}) - \Ent(\txv_{t+1}|\txm_t,\rxm_{t+1})\\
    \stackrel{\rm{( b)}}{=}&\Ent(\txv_{t+1}|\txm_{t-1},\rxm_{t-1},\rxv_{t+1}) - \Ent(\txv_{t+1}|\txm_t,\rxm_{t+1})
    \stackrel{\rm{( c)}}{\geq} 0.
\end{align*}
Here, $^{(a)}$ uses the independence between $\txv_t$ and $(\txm_{t-1},\rxm_{t-1})$, $^{(b)}$ uses Assumption A1, $^{(c)}$ uses conditioning to reduce entropy. Thus, $\MuI(\txv_t;\rxv_t|\txm_{t-1},\rxm_{t-1})$ is monotonically increasing with $t$ for all $t>\Tt$. 
Then, in the limit when $\Tlim\to\infty$, we get  \eqref{eq:MuI_begin_lower_bound}.

Also, 
 \begin{align*}
     &\MuI(\txv_{\Tt+1};\rxv_{\Tt+1}|\txm_{\Tt},\rxm_{\Tt})\\
     =&\Ent(\txv_{\Tt+1})-\Ent(\txv_{\Tt+1}|\txm_{\Tt},\rxm_{\Tt},\rxv_{\Tt+1})\\
     \stackrel{\rm{( a)}}{=}&\Ent(\txv_{\Tt+1})-\Ent(\txv_{\Tt+1}|\txm_{\Tt},\rxm_{\Tt},\hat{\chv}_{\Tb},\rxv_{\Tt+1})\\
     \stackrel{\rm{( b)}}{\geq}&\Ent(\txv_{\Tt+1})-\Ent(\txv_{\Tt+1}|\hat{\chv}_{\Tb},\rxv_{\Tt+1})\\
     =&\MuI(\txv_{\Tt+1};\rxv_{\Tt+1}|\hat{\chv}_{\Tb}),
 \end{align*}
 where $^{(a)}$ uses that $\hat{\chv}_{\Tb}$ is a function of $(\txm_{\Tt},\rxm_{\Tt})$, and $^{(b)}$ uses conditioning to reduce entropy. By taking the limit $\Tlim\to \infty$, we have \eqref{eq:MuI_estimate_lower_bound}.

\subsection{Proof of  \eqref{eq:MuI_from_derivative}}
We first show the derivative relationship between $\EntInn(\rxA_\dy)$ and $\EntR(\rxA_\dy)$ defined below, and then generalize to the conditional entropies which directly lead to the conclusion \eqref{eq:MuI_from_derivative}.
Define
\begin{align*}
\EntInn(\rxA_\dy) =\lim_{\Tlim\to\infty}\Ent(\rxv_{\ceil{\dy\Tt}+1}|\rxm_{\ceil{\dy\Tt}}),
    \numberthis
    \label{eq:EntInn_def}
\end{align*}
\begin{equation}
    \EntR(\rxA_\dy)=\lim_{\Tlim\to\infty}\frac{1}{\Tt}\Ent(\rxm_{\ceil{\dy\Tt}}),
    \label{eq:EntR_def}
\end{equation}
which can be considered as  $\EntInnC{\txA}{\rxA}{\dx}{\dy}$ and $\EntR(\rxA_\dy|\txA_{\dx})$ with $\dx=0$. For mathematical rigorousness, we keep the $\ceil{\cdot}$ notation here.
We show that, under some conditions, 
 $\EntInn(\rxA_\dy)$ is the derivative of $\EntR(\rxA_\dy)$.
\begin{thm}
\label{thm:derivative_monotonic}
Suppose there exists a $\kappa>0$ so that $\Ent(\rxv_{t+1}|\rxm_{t})$ is monotonic in $t$ when $t\in[\floor{(\dy-\kappa)\Tt},\ceil{(\dy+\kappa)\Tt}]$ as $\Tlim\to\infty$. 

If $\EntR(\rxA_\dy)$ and its derivative with respect to $\dy$ exist, we have
\begin{equation}
    \EntInn(\rxA_\dy) =\frac{\partial\EntR(\rxA_\dy)}{\partial\dy}.
    \numberthis
    \label{eq:derivative_2}
\end{equation}
If both $\EntR(\rxA_\dy)$ and $\EntInn(\rxA_\dy)$ exist, and there exists a $c>0$ independent of $t$ and $\Tt$ so that $|\Ent(\rxv_{t+1}|\rxm_{t})|<c$,  we have
\begin{equation}
    \EntR(\rxA_\dy) = \int_{0}^{\dy}\EntInn(\rxA_u) d u .
    \label{eq:theorem_integral_monotonic_thm}
\end{equation}
\end{thm}

\begin{proof}
Equation \eqref{eq:theorem_integral_monotonic_thm} is an integral equivalent of \eqref{eq:derivative_2} and
we only prove \eqref{eq:derivative_2} for simplicity.  Without loss of generality, we assume that $\Ent(\rxv_{t+1}|\rxm_{t})$ is monotonically decreasing. Using the definition of $\EntR(\rxA_\dy)$ in \eqref{eq:EntR_def}, we have
\begin{align*}
    &\frac{1}{\kappa}(\EntR(\rxA_{\dy+\kappa})-\EntR(\rxA_\dy))\\
    =&\lim_{\Tlim\to\infty}\frac{\Ent(\rxm_{\ceil{(\dy+\kappa)\Tt}})-\Ent(\rxm_{\ceil{\dy\Tt}})}{\kappa\Tt}\\
    =&\lim_{\Tlim\to\infty}\frac{\sum_{t=\ceil{\dy\Tt}+1}^{\ceil{(\dy+\kappa)\Tt}}\Ent(\rxv_{t}|\rxm_{t-1})}{\kappa\Tt}\\
    \leq &\lim_{\Tlim\to\infty}\frac{(\kappa\Tt+1)\cdot\Ent(\rxv_{\ceil{\dy\Tt}+1}|\rxm_{\ceil{\dy\Tt}})}{\kappa\Tt}\\
    =&\lim_{\Tlim\to\infty}\Ent(\rxv_{\ceil{\dy\Tt}+1}|\rxm_{\ceil{\dy\Tt}}).
    \numberthis
    \label{eq:limit_upper_bound}
\end{align*}

Similarly to \eqref{eq:limit_upper_bound}, we also have
\begin{align*}
    &\lim_{\Tlim\to\infty}\Ent(\rxv_{\ceil{\dy\Tt}+1}|\rxm_{\ceil{\dy\Tt}})
    \leq\frac{1}{\kappa}(\EntR(\rxA_\dy)-\EntR(\rxA_{\dy-\kappa})).
    \numberthis
    \label{eq:limit_lower_bound}
\end{align*}
Let $\kappa\downto 0$ in both \eqref{eq:limit_upper_bound} and \eqref{eq:limit_lower_bound}; because we assume that the derivative of $\EntR(\rxA_\dy)$ exists, these limits both equal this derivative.  Then, the definition of $\EntInn(\rxA_\dy)$ in \eqref{eq:EntInn_def} yields \eqref{eq:derivative_2}.
\end{proof}

Theorem \ref{thm:derivative_monotonic} is a consequence of the entropy chain rule and letting an infinite sum converge to an integral (standard Riemann sum approximation).  Such an analysis has also been used in the context of computing mutual information; for example  \cite{shamai2001impact,guo2005randomly,guo2005mutual,guo2008multiuser,honig2009advances}. 
Theorem \ref{thm:derivative_monotonic}  can be generalized to include conditioning on $\txA$, thus leading to the following corollary, provided that $\EntR(\rxA_\dy|\txA_\dx)$ and its derivative with respect to $\dy$ exist.
\begin{cor}
\label{cor:A1_derivative_relation}
Assume A1 holds.  (a) For $\dy>1$,
\begin{align*}
    \EntInnC{\txA}{\rxA}{}{\dy} =\frac{\partial\EntR(\rxA_\dy|\txA)}{\partial\dy},
    \numberthis
    \label{eq:derivative_data}
\end{align*}
\begin{align*}
    \EntInnC{\txA}{\rxA}{\dy^+}{\dy} =\frac{\partial\EntR(\rxA_\dy|\txA_\dy)}{\partial\dy}.
    \numberthis
    \label{eq:derivative_training}
\end{align*}

(b) If $\txv_t$ are \iid for all $t$, then for all $\dy,\dx>0$ and $\dy\neq\dx$,
\begin{align*}
    \EntInnC{\txA}{\rxA}{\dx}{\dy} &=\frac{\partial\EntR(\rxA_\dy|\txA_\dx)}{\partial\dy},
    \numberthis
    \label{eq:derivative_more_iid_training} \\
    \EntInnC{\txA}{\rxA}{^+}{} &=\left.\frac{\partial\EntR(\rxA_\dy|\txA_\dy)}{\partial\dy}\right|_{\dy=1}.
    \numberthis
    \label{eq:derivative_training_0}
\end{align*}

(c) If
\begin{align}
    \lim_{\dy\downto 1}\EntInnC{\txA}{\rxA}{}{\dy}=\lim_{\dx\upto 1}\EntInnC{\txA}{\rxA}{\dx}{},
    \label{eq:two_bound_met}
\end{align}
then Assumption A2 \eqref{eq:continue_in_data} is met.
\end{cor}

\begin{proof} (a) Under A1, for all $\dx\geq 1$, we have
\begin{align*}
    \Ent(\rxv_{t+1}|\txm_{\ceil{\dx\Tt}},&\rxm_{t})\leq\Ent(\rxv_{t+1}|\txm_{\ceil{\dx\Tt}},\rxm_{t-1})\\
    =&\;\Ent(\rxv_{t}|\txm_{\ceil{\dx\Tt}},\rxm_{t-1}),
    \numberthis
    \label{eq:monotonic_A1}
\end{align*}
when $\ceil{\Tt}+1\leq t\leq \ceil{\dx\Tt}-1$ or $t\geq \ceil{\dx\Tt}+1$. Here, we use that the input is \iid and the system is memoryless and time invariant; the inequality follows from the fact that conditioning reduces entropy.  
Therefore, $\forall \kappa\in(0,\dy-1)$, $\Ent(\rxv_{t+1}|\txm_{\ceil{\Tt}},\rxm_{t})$ is monotonically decreasing in $t$ for $t\in[\floor{(\dy-\kappa)\Tt},\ceil{(\dy+\kappa)\Tt}]$ when $\Tt>\frac{2}{\dy-1-\kappa}$. Then, Theorem \ref{thm:derivative_monotonic} yields \eqref{eq:derivative_data}.

Also, $\forall \kappa\in(0,\dy-1),\dx>2\dy-1$, $\Ent(\rxm_{t+1}|\txm_{\ceil{\dx\Tt}},\rxm_{t})$ is monotonically decreasing in $t$ for $t\in[\floor{(\dy-\kappa)\Tt},\ceil{(\dy+\kappa)\Tt}]$ when $\Tt>\max(\frac{2}{\dy-1-\kappa},\frac{2}{\dx-\dy-\kappa})$. Then, Theorem \ref{thm:derivative_monotonic} yields
\begin{equation}
    \EntInnC{\txA}{\rxA}{\dx}{\dy} =\frac{\partial\EntR(\rxA_\dy|\txA_{\dx})}{\partial\dy}.
    \label{eq:derivative_1}
\end{equation}
Assumption A1 yields
\begin{align}
    \EntInnC{\txA}{\rxA}{\dx}{\dy}=\EntInnC{\txA}{\rxA}{\dy^+}{\dy},
    \label{eq:H_prime_equ}
\end{align}
where $\EntInnC{\txA}{\rxA}{\dy^+}{\dy}$ is defined in \eqref{eq:EntInn_cond_def_one_extra_high}, and 
\begin{align}
    \EntR(\rxA_\dy|\txA_{\dx})=\EntR(\rxA_\dy|\txA_{\dy}).
    \label{eq:H_equ}
\end{align}
Therefore, \eqref{eq:derivative_1} becomes \eqref{eq:derivative_training}.

(b) If $\txv_t$ are \iid for all $t$, then \eqref{eq:monotonic_A1} is valid for all $t\leq \ceil{\dx\Tt}-1$ or $t\geq \ceil{\dx\Tt}+1$. Therefore, Theorem \ref{thm:derivative_monotonic} yields \eqref{eq:derivative_more_iid_training}. By taking $\dy=1$ and $\dx>1$, \eqref{eq:derivative_more_iid_training}, \eqref{eq:H_prime_equ}, and \eqref{eq:H_equ} then yield \eqref{eq:derivative_training_0}.

(c) For $t\geq\ceil{\Tt}+1$, we have
\begin{align*}
  \Ent(\rxv_{t+1}|\txm_{\ceil{\Tt}},\rxm_{t}) \leq\Ent(\rxv_t|\txm_{\ceil{\Tt}},\rxm_{t-1}).
\end{align*}
Therefore,
\begin{equation*}
   \lim_{\dy\downto 1}\EntInnC{\txA}{\rxA}{}{\dy}\leq \EntInnC{\txA}{\rxA}{}{}.
\end{equation*}
Conditioning to reduce entropy again yields
\begin{align*}
  \Ent(\rxv_{\ceil{\Tt}+1}|\txm_{\ceil{\Tt}},\rxm_{\ceil{\Tt}}) \leq&\; \Ent(\rxv_{\ceil{\Tt}+1}|\txm_{\ceil{\dx\Tt}},\rxm_{\ceil{\Tt}}),
\end{align*}
for any $\dx<1$ and therefore,
\begin{equation*}
    \EntInnC{\txA}{\rxA}{}{}\leq \lim_{\dx\upto 1}\EntInnC{\txA}{\rxA}{\dx}{}.
\end{equation*}
Equation \eqref{eq:two_bound_met} then implies A2 \eqref{eq:continue_in_data}.
\end{proof}

Corollary \ref{cor:A1_derivative_relation} can now be used to finish the proof of Theorem \ref{thm:MI_equal_gap&computation_of_MuI_from_derivative}. By \eqref{eq:cI_def_with_P_tau_high}, and Assumptions A1--2
\begin{equation*}
    \cIInn{\txA}{\rxA}{} =\EntInnC{\txA}{\rxA}{}{} - \EntInnC{\txA}{\rxA}{^+}{} = \lim_{\dy\downto 1} \EntInnC{\txA}{\rxA}{}{\dy} - \lim_{\dy\downto 1} \EntInnC{\txA}{\rxA}{\dy^+}{\dy},
\end{equation*}
and together with Corollary \ref{cor:A1_derivative_relation}(a), we have \eqref{eq:MuI_from_derivative}.

\bibliographystyle{IEEEtran}
\bibliography{bib/IEEEarbv,bib/refs.bib}

\end{document}

%% file: src/header.tex
\usepackage{epsfig,epsf}
\usepackage{color}

\usepackage{bbold}
\usepackage{multirow}
\usepackage{latexsym}
\usepackage{amssymb}
\usepackage{amsmath}

\usepackage{amsfonts}
\usepackage[sort]{cite}
\usepackage{fullpage}
\usepackage{mathrsfs}
\usepackage{scalefnt}
\usepackage{extarrows}
\usepackage{setspace}
\usepackage{stfloats}
\usepackage{url}
\usepackage{hyphenat}
\usepackage{graphicx}
\usepackage[table]{xcolor}
\usepackage{bbm}
\usepackage{dsfont}
\usepackage{bm}
\usepackage{array}
\usepackage{wrapfig}
\usepackage[]{algpseudocode}
\usepackage[linesnumbered,ruled]{algorithm2e}

\let\oldnl\nl
\newcommand{\nonl}{\renewcommand{\nl}{\let\nl\oldnl}}

\usepackage{mathtools}
\allowdisplaybreaks[1]
\usepackage{epstopdf}
\newcommand\numberthis{\addtocounter{equation}{1}\tag{\theequation}}

\newsavebox\myboxA
\newsavebox\myboxB
\newlength\mylenA
\newcommand*\xoverline[2][0.75]{%
    \sbox{\myboxA}{$\m@th#2$}%
    \setbox\myboxB\null
    \ht\myboxB=\ht\myboxA%
    \dp\myboxB=\dp\myboxA%
    \wd\myboxB=#1\wd\myboxA
    \sbox\myboxB{$\m@th\overline{\copy\myboxB}$}
    \setlength\mylenA{\the\wd\myboxA}
    \addtolength\mylenA{-\the\wd\myboxB}%
    \ifdim\wd\myboxB<\wd\myboxA%
       \rlap{\hskip 0.5\mylenA\usebox\myboxB}{\usebox\myboxA}%
    \else
        \hskip -0.5\mylenA\rlap{\usebox\myboxA}{\hskip 0.5\mylenA\usebox\myboxB}%
    \fi}
\makeatother  
\usepackage{enumerate}

\def\phi{\varphi}

\def\tendsto{\rightarrow}

\def\Tp{{\intercal}}

\newcommand{\comm}[1]{}

\def\bff{{\mathbf{f}}}
\def\bg{{\mathbf{g}}}

\def\bk{{\mathbf{k}}}

\def\bv{{\mathbf{v}}}

\def\bx{{\mathbf{x}}}
\def\by{{\mathbf{y}}}
\def\bz{{\mathbf{z}}}
\def\b0{{\mathbf{0}}}

\def\argmax{\mathop{\mathrm{argmax}}}

\def\bQ{{\mathbf{Q}}}

\def\cH{\mathcal{H}}
\def\cI{\mathcal{I}}

\def\cN{\mathcal{N}}

\def\cP{\mathcal{P}}

\def\cR{\mathcal{R}}

\def\cX{\mathcal{X}}
\def\cY{\mathcal{Y}}

\def \Re[#1]{\text{Re}\left(#1\right)}
\def \Im[#1]{\text{Im}\left(#1\right)}
\def \Cpx[#1]{\tilde{#1}}
\def \tx {x} 
\def \rx {y} 

\def \txv {\bx} 
\def \rxv {\by} 

\def \txm {X} 
\def \rxm {Y} 

\def \txA {\cX} 
\def \rxA {\cY} 
\def \chm {G} 

\def \txvte[#1]{\txv_{\rm t,#1}} 
\def \rxvte[#1]{\rxv_{\rm t,#1}} 
\def \txvde[#1]{\txv_{\rm d,#1}} 
\def \rxvde[#1]{\rxv_{\rm d,#1}} 
\def \nvbe[#1]{\nv_{\rm b,#1}} 

\def \txde[#1]{\tx_{\rm d,#1}} 
\def \rxde[#1]{\rx_{\rm d,#1}} 

\def \txta[#1]{\tx_{{\rm t},#1}} 
\def \txda[#1]{\tx_{{\rm d},#1}} 

\def \txvta[#1]{\txv_{{\rm t},#1}} 
\def \rxvta[#1]{\rxv_{{\rm t},#1}} 
\def \txvda[#1]{\txv_{{\rm d},#1}} 
\def \rxvda[#1]{\rxv_{{\rm d},#1}} 
\def \nvba[#1]{\nv_{{\rm d},#1}} 
\def \rxta[#1]{\rx_{{\rm t},#1}} 
\def \txda[#1]{\tx_{{\rm d},#1}} 
\def \rxda[#1]{\rx_{{\rm d},#1}}

\def \ptrve[#1]{\hat{p}_{(\txvte[#1],\rxvte[#1])}}

\def \onev[#1]{\underline{1}_{#1}}

\def \MuI {I} 
\def \Ent {H} 
\def \EntR {\cH}
\def \EntInn {{\cH{'}}}

\def \E {\mathbb{E}} 
\def \Pr {{P}} 

\def \chv {\bg} 

\def \chs {g} 
\def \tn {M} 
\def \rn {N} 
\def \ns  {v}  

\def \nv {\bv} 

\def \Tb {T} 
\def \Tt {{\tau T}} 
\def \Tlim {\Tb} 

\def \sign {\text{sign}}

\DeclarePairedDelimiter{\ceil}{\lceil}{\rceil}
\DeclarePairedDelimiter{\floor}{\lfloor}{\rfloor}

\def \ratio {\alpha}

\def \qha[#1]{q_{{ \chs},#1}}
\def \Qha[#1]{Q_{{\rm \chs},#1}}

\def \qxa[#1]{q_{{ x},#1}}
\def \Qxa[#1]{Q_{{\rm x},#1}}

\def \bQha[#1]{\bQ_{{\rm h},#1}}
\def \bQxa[#1]{\bQ_{{\rm x},#1}}

\def \Ropt {\cR_{\rm opt}}

\def \dy {\varepsilon}
\def \dx {\delta}

\def \iid {{\it iid}\xspace}

\def \rxdequa[#1]{\tilde{\rx}_{{\rm d},#1}}

\def \htxta[#1]{\hat{\tx}_{{\rm t},#1}} 
\def \htxda[#1]{\hat{\tx}_{{\rm d},#1}} 
\def \ttxta[#1]{\tilde{\tx}_{{\rm t},#1}} 
\def \ttxda[#1]{\tilde{\tx}_{{\rm d},#1}} 
\def \zta[#1]{z_{{\rm t},#1}}
\def \bzda[#1]{\bz_{{\rm d},#1}}
\def \bzta[#1]{\bz_{{\rm t},#1}}
\def \zda[#1]{z_{{\rm d},#1}}

\def \rxvda[#1]{\rxv_{{\rm d},#1}}

\def \expeq[#1]{\overset{{#1}}{\equiv}}

\def \Enta[#1]{\Omega_{#1}}
\def \vara[#1]{\sigma^2_{#1}}

\def \tauopt{\tau_{{\rm opt}}}
\def \upto {{\nearrow}}
\def \downto {{\searrow}}

\def \EntInnC#1#2#3#4{\EntInn({#2}_{#4}|{#1}_{#3})}

\def \EntInnY#1#2{\EntInn({#1}_{#2})}

\def \AEntInnY#1#2{\bar{\EntInn}({#1}_{#2})}

\def \Pe {P_{\rm e}}

\def \cIInn#1#2#3{\cI{'}({#1}_{#3};{#2}_{#3})}
\def \cII#1#2#3{\cI{'}({#1};{#2}{#3})}

\def \cIRA#1#2#3{\cI({#1};{#2}{#3})}

\def \chva#1{\chv_{#1}}

\usepackage{amsthm}

\newtheorem{thm}{{Theorem}}

\newtheorem{cor}{Corollary}

